\documentclass[a4paper]{article}

\usepackage{amssymb}
\usepackage{microtype}
\usepackage{url}
\usepackage{amsmath}
\usepackage{courier}
\usepackage{graphicx}
\usepackage{multirow}
\usepackage{verbatim}

\newcommand{\BIGCOMMENT}[1]{}

\newcommand{\COMMENT}[1]{}

\usepackage{fullpage,latexsym,amsthm,amsfonts,amssymb,amsmath,amsthm}

\newtheorem{definition}{Definition}
\newtheorem{theorem}[definition]{Theorem}
\newtheorem{claim}{Claim}
\newtheorem{fact}{Fact}
\newtheorem{corollary}[definition]{Corollary}
\newtheorem{lemma}[definition]{Lemma}

\begin{document}

\title{New Bounds for the Garden-Hose Model\footnote{This work is funded by the Singapore Ministry of Education (partly through the Academic Research Fund Tier 3 MOE2012-T3-1-009) and by the Singapore National Research Foundation.}
}

\author{Hartmut Klauck\\ CQT and Nanyang Technological University\\ Singapore\\ \texttt{hklauck@gmail.com}
\and Supartha Podder\\ CQT\\ Singapore\\ \texttt{supartha@gmail.com}}

\date{}
\maketitle

\begin{abstract}
We show new results about the garden-hose model.
Our main results include improved lower bounds based on non-deterministic communication complexity (leading to the previously unknown $\Theta(n)$ bounds for Inner Product mod 2 and Disjointness), as well as an $O(n\cdot \log^3 n)$ upper bound for the Distributed Majority function (previously conjectured to have quadratic complexity). We show an efficient simulation of formulae made of AND, OR, XOR gates in the garden-hose model, which implies that lower bounds on the garden-hose complexity $GH(f)$ of the order $\Omega(n^{2+\epsilon})$ will be hard to obtain for explicit functions.
Furthermore we study a time-bounded variant of the model, in which even modest savings in time can lead to exponential lower bounds on the size of garden-hose protocols.
\end{abstract}

\section{Introduction}

\subsection{Background: The Model}
Recently, Buhrman et al.~\cite{BFSS11} proposed a new measure of complexity for finite Boolean functions, called {\em garden-hose complexity}.
This measure can be viewed as a type of distributed space complexity, and while its motivation is mainly in
applications to position based quantum cryptography, the playful definition of the model is quite appealing in itself. Garden-hose complexity can be viewed as a natural measure of space, in a situation where two players with private inputs compute a Boolean function cooperatively.
Space-bounded communication complexity has been investigated before \cite{beame:commtrade,klauck:tradeoffs,ksw:dpt-siam} (usually for problems with many outputs), and recently Brody et al.~\cite{BrodyCPSS13} have studied a related model of space bounded communication complexity for Boolean functions (see also \cite{Overlays}). In this context the garden-hose model can be viewed as a memoryless model of communication that is also reversible.

To describe the garden-hose model let us consider two neighbors, Alice and Bob. They own adjacent gardens which happen to have $s$ empty water pipes crossing their common boundary. These pipes are the only means of communication available to the two. Their goal is to compute a Boolean function on a pair of private inputs, using water and the pipes across their gardens as a means of communication\footnote{It should be mentioned that even though Alice and Bob choose to not communicate in any other way, their intentions are not hostile and neither will deviate from a previously agreed upon protocol.}.

A garden-hose protocol works as follows: There are $s$ shared pipes.
Alice takes some pieces of hose and connects pairs of the open ends of the $s$ pipes.
She may keep some of the ends open. Bob acts in the same way for his end of the pipes. The connections Alice and Bob place depend on their local inputs $x,y$, and we stress that every end of a pipe is only connected to at most one other end of a pipe (meaning no Y-shaped pieces of hose may be used to split or combine flows of water).
Finally, Alice connects a water tap to one of those open ends on her side and starts the water.
Based on the connections of Alice and Bob, water flows back and forth through the pipes and finally ends up spilling on one side.

If the water spills on Alice's side we define the output to be 0. Otherwise, the water spills on Bob's side and the output value is $1$. It is easy to see that due to the way the connections are made the water must eventually spill on one of the two sides, since cycles are not possible.

Note that the pipes can be viewed as a communication channel that can transmit $\log s$ bits, and that the garden-hose protocol is memoryless, i.e., regardless of the previous history, water from pipe $i$ always flows to pipe $j$ if those two pipes are connected. Furthermore computation is reversible, i.e., one can follow the path taken by the water backwards (e.g. by sucking the water back).

Buhrman et al.~\cite{BFSS11} have shown that it is possible to compute every function $f:\{0,1\}^n\times\{0,1\}^n\to\{0,1\}$ by playing a garden-hose game.
A garden-hose protocol consists of the scheme by which Alice chooses her connections depending on her private input $x \in \{0,1\}^n$ and how Bob chooses his connections depending on his private input $y \in \{0,1\}^n$. Alice also chooses the pipe that is connected to the tap. The protocol computes a function $f$, if for all inputs with $f(x,y)=0$ the water spills on Alice's side, and for all inputs with $f(x,y)=1$ the water spills on Bob's side.

The size of a garden-hose protocol is the number $s$ of pipes used. The garden-hose complexity GH($f$) of a function $f(x,y)$ is the minimum number of pipes needed in any garden-hose game that computes the value of $f$ for all $x$ and $y$ such that $f(x,y)$ is defined.

The garden-hose model is originally motivated by an application to quantum position-verification schemes \cite{BFSS11}. In this setting the position of a prover is verified via communications between the prover and several verifiers. An attack on such a scheme is performed by several provers, none of which are in the claimed position. \cite{BFSS11} proposes a protocol for position-verification that depends on a function $f:\{0,1\}^n\times\{0,1\}^n\to\{0,1\}$, and a certain attack on this scheme requires the attackers to share as many entangled qubits as the garden-hose complexity of $f$. Hence all $f$ with low garden-hose complexity are not suitable for this task, and it becomes desirable to find explicit functions with large garden-hose complexity.

Buhrman et al.~\cite{BFSS11} prove a number of results about the garden-hose model:

\begin{itemize}
\item Deterministic one-way communication complexity can be used to show lower bounds of up to $\Omega(n/\log n)$ for many functions.
\item For the Equality problem they refer to a bound of $GH(Equality)=\Theta(n)$ shown by Pietrzak (the proof  implicitly uses the fooling set technique from communication complexity \cite{KN96} [personal communication]).
\item They argue that super-polynomial lower bounds for the garden-hose complexity of a function $f$ imply that the function cannot be computed in Logspace, making such bounds hard to prove for `explicit' functions.
    \item They define randomized and quantum variants of the model and show that randomness can be removed at the expense of multiplying size by a factor of $O(n)$ (for quantum larger gaps are known).
        \item Via a counting argument it is easy to see that most Boolean functions need size $GH(f)=2^{\Omega(n)}$.
\end{itemize}

Very recently Chiu et al.~\cite{Ch13} have improved the upper bound for the Equality function to $1.359n$ from the previously known $2n$ bound \cite{BFSS11}.

\subsection{Our Results}
We study garden-hose complexity and establish several new connections with well studied models like communication complexity, permutation branching programs, and formula size.

We start by showing that non-deterministic communication complexity gives lower bounds on the garden-hose complexity of any function $f$. This improves the lower bounds of $\Omega(\frac{n}{\log n})$ for several important functions like Inner Product, Disjointness to $\Omega(n)$.

We observe that any 2-way deterministic communication protocol can be converted to a garden-hose protocol so that the complexity $GH(f)$ is upper bounded by the {\em size} of the protocol tree of the communication protocol.

We then turn to comparing the model to another nonuniform notion of space complexity, namely branching programs. We show how to convert any {\em permutation branching program} to a garden-hose protocol with only a constant factor loss in size.

The most important application of this simulation is that it allows us to find a garden-hose protocol for the distributed Majority function,  \mbox{$DMAJ(x,y)=1$ iff $\sum^{n}_{i=1} (x_i \cdot y_i) \geq \frac{n}{2}$}, that has size $O(n \cdot \log^3 n)$, disproving the conjecture in \cite{BFSS11} that this function has complexity $\Omega(n^2)$.

Using the garden-hose protocols for Majority, Parity, AND, OR, we show upper bounds on the composition of functions with these.

 We then show how to convert any Boolean formula with AND, OR, XOR gates to a garden-hose protocol with a small loss in size. In particular, any formula consisting of arbitrary fan-in 2 gates only can be simulated by a garden-hose protocol with a constant factor loss in size. This result strengthens the previous observation that explicit super-polynomial lower bounds for  $GH(f)$ will be hard to show: even bounds of $\Omega(n^{2+\epsilon})$ would improve on the long-standing best lower bounds on formula size due to Ne\v ciporuk\ from 1966 \cite{Ne66}. We can also simulate formulae including a limited number of Majority gates of arbitrary fan-in, so one might be worried that even super-linear lower bounds could be difficult to prove. We argue, however, that for formulae using arbitrary symmetric gates we can still get near-quadratic lower bounds using a Ne\v ciporuk-type method. Nevertheless we have to leave super-linear lower bounds on the garden-hose complexity as an open problem.

 Next we define a notion of {\it time} in garden-hose protocols and prove that for any function $f$, if we restrict the number of times water can flow through pipes to some value $k$, we have $GH_k(f) = \Omega(2^{D_k(f)/k})$, where $GH_k$ denotes the time-bounded garden-hose complexity, and $D_k$ the $k$-round deterministic communication complexity. This result leads to strong lower bounds for the time bounded complexity of e.g.~Equality, and to a time-hierarchy based on the pointer jumping problem.

Finally, we further investigate the power of randomness in the garden-hose model by considering private coin randomness (\cite{BFSS11} consider only public coin randomness).

\subsection{Organization}

Most proofs are deferred to the appendix.

\section{Preliminaries}

\subsection{Definition of the Model}
We now describe the garden-hose model in graph terminology.
In a garden-hose protocol with $s$ pipes
there is a set $V$ of $s$ vertices plus one extra vertex, the {\em tap} $t$.

Given their inputs $x,y$ Alice and Bob want to compute $f(x,y)$. Depending on $x$ Alice connects some of the vertices in $V\cup\{t\}$ in pairs by adding edges $E_A(x)$ that form a matching among the vertices in $V\cup \{t\}$.
Similarly Bob connects some of the vertices in $V$ in pairs by adding edges $E_B(y)$ that form a matching in $V$.

Notice that after they have added the additional edges, a path starting from vertex $t$ is formed in the graph $G=(V\cup\{t\}, E_A(x)\cup E_B(y))$. Since no vertex has degree larger than 2, this path is unique and ends at some vertex. We define the output of the game to be the parity of the length of the path starting at $t$. For instance, if the tap is not connected the path has length 0, and the output is 0. If the tap is connected to another vertex, and that vertex is the end of the path, then the path has length 1 and the output is 1 etc.

A garden-hose protocol for $f:{\cal X}\times {\cal Y}\to \{0,1\}$ is a mapping from $x\in \cal X$ to matchings among $V\cup \{t\}$ together with a mapping from $y\in\cal Y$ to matchings among $V$. The protocol computes $f(x,y)$ if for all $x,y$ the path has even length iff $f(x,y)=0$. The garden-hose complexity of $f$ is the smallest $s$ such that a garden-hose protocol of size $s$ exists that computes $f$.

We note that one can form a matrix $G_s$ that has rows labeled by all of Alice's matchings, and columns labeled by Bob's matchings, and contains the parity of the path lengths. A function $f$ has garden-hose complexity $s$ iff its communication matrix is a sub-matrix of $G_s$. $G_s$ is called the {\em garden-hose matrix} for size $s$.

\subsection{Communication Complexity, Formulae, Branching Programs}

\begin{definition}
 Let $f: \{0,1\}^n \times \{0,1\}^n \rightarrow \{0,1\}$. In a communication complexity protocol two players Alice and Bob receive inputs $x$ and $y$ from $\{0,1\}^n$. In the protocol players exchange messages in order to compute $f(x,y)$.
 Such a protocol is represented by a protocol tree, in which vertices, alternating by layer, belong to Alice or to Bob, edges are labeled with messages, and leaves either accept or reject. See \cite{KN96} for more details. The communication matrix is the matrix containing $f(x,y)$ in row $x$ and column $y$.

  We say a protocol $P$ correctly computes the function $f(x,y)$ if for all $x$, $y$ the output of the protocol $P(x,y)$ is equal to $f(x,y)$. The communication complexity of a protocol is the maximum number of bits exchanged for all $x, y$.

 The deterministic communication complexity $D(f)$ of a function $f$ is the complexity of an optimal protocol that computes $f$.

 \end{definition}

 \begin{definition}
The non-deterministic communication complexity $N(f)$ of a Boolean function $f$ is the
length of the communication in an optimal two-player protocol in which Alice and Bob can make non-deterministic guesses,
and there are three possible outputs $\tt accept, reject, undecided$.
For each $x,y$ with $f(x,y)=1$ there is a guess that will make the players accept but there is no guess that will make the players reject, and vice versa for inputs with $f(x,y)=0$.
\end{definition}

Note that the above is the two-sided version of non-deterministic communication complexity. It is well known \cite{KN96} that $N(f)\leq D(f)\leq O(N^2(f))$, and that these inequalities are tight.

\begin{definition}
In a public coin randomized protocol for $f$ the players have access to a public source of random bits. For all inputs $x,y$ it is required that the protocol gives the correct output with probability $1 - \epsilon$ for some $\epsilon < 1/2$. The public coin randomized communication complexity of $f$, $R^{pub}_{\epsilon}(f)$ is the complexity of the optimal public coin randomized protocol. Private coin protocols are defined analogously (players now have access only to private random bits), and their complexity is denoted by $R_\epsilon(f)$.
\end{definition}

\begin{definition} The deterministic communication complexity of protocols with at most $k$ messages exchanged, starting with Alice,  is denoted by $D_k(f)$.
\end{definition}

\begin{definition}  In a simultaneous message passing protocol, both Alice and Bob send messages $m_A, m_B$ to a referee. The referee, based on $m_A, m_B$, computes the output.
The simultaneous communication complexity of a function $f$, $R^{||}(f)$, is the cost of the best simultaneous protocol that computes the function $f$ using private randomness and error 1/3.\end{definition}

Next we define Boolean formulae.

\begin{definition}
 A Boolean formula is a Boolean circuit whose every node has fan-out 1 (except the output gate). A Boolean formula of depth $d$ is then a tree of depth $d$. The nodes are labeled by gate functions from a family of allowed gate functions, e.g. the class of
 the 16 possible functions of the form $f:\{0,1\} \times \{0,1\} \rightarrow \{0,1\}$ in case the fan-in is restricted to 2. Another interesting class of gate functions is the class of all symmetric functions (of arbitrary fan-in). The {\em formula size} of a function $f$ (relative to a class of gate functions) is the smallest number of leaves in a formula computing $f$.
\end{definition}

Finally, we define branching programs. Our definition of permutation branching programs is extended in a slightly non-standard way.

\begin{definition}
 A branching program is a directed acyclic graph with one source node and two sink nodes (labeled with $\tt accept$ and $\tt reject$). The source node has in-degree 0. The sink nodes have out-degree 0. All non-sink nodes are labeled by  variables $x_i \in \{x_1,\cdots,x_n \}$ and have out-degree 2. The computation on an input $x$ starts from the source node and depending on the value of $x_i$ on a node either moves along the left outgoing edge or the right outgoing edge of that node. An input $x \in \{0,1\}^n$ is accepted iff the path defined by $x$ in the branching program leads to the sink node labeled by $\tt accept$. The length of the branching program is the maximum length of any path,  and the size is the number of nodes.

 A layered branching program of length $l$ is a branching program where all non-sink nodes (except the source) are partitioned into $l$ layers. All the nodes in the same layer query the same variable $x_i$, and all outgoing edges of the nodes in a layer go to the nodes in the next layer or directly to a sink. The width of a layered branching program is defined to be the maximum number of nodes in any layer of the program. We consider the starting node to be in layer 0 and the sink nodes to be in layer $l$.

 A permutation branching program is a layered branching program, where each layer has the same number $k$ of nodes, and
 if $x_i$ is queried in layer $i$, then the edges labeled with 0 between layers $i$ and $i+1$ form an injective mapping from $\{1,\ldots,k\}$ to $\{1,\ldots, k\}\cup\{\tt accept, reject\}$ (and so do the the edges labeled with 0). Thus, for permutation branching programs if we fix the value of $x_i$, each node on level $i+1$ has in-degree at most 1.

We call a permutation branching program {\em strict} if there are no edges to $\tt accept/reject$ from internal layers. This is the original definition of permutation branching programs. Programs that are not strict are also referred to as {\em loose} for emphasis.

We denote by $PBP(f)$ the minimal size of a permutation branching program that computes $f$.
\end{definition}

We note that simple functions like AND, OR can easily be computed by linear size loose permutation branching programs of width 2, something that is not possible for strict permutation branching programs \cite{bar85}.

\section{Garden-Hose Protocols and Communication Complexity}
\subsection{Lower Bound via Non-deterministic Communication}

In this section we show that non-deterministic communication complexity can be used to lower bound $GF(f)$. This bound is often better than the bound $GH(f)\geq\Omega(D_1(f)/\log(D_1(f)))$ shown in  \cite{BFSS11}, which cannot be larger than $n/\log n$.

\begin{theorem}\label{thm:ncc} $GH(f)\geq N(f)-1$.
\end{theorem}

The main idea is that a nondeterministic protocol that simulates the garden-hose game can choose the {\em set} of pipes that are used on a path used on inputs $x,y$ instead of the path itself, reducing the complexity of the protocol. The set that is guessed may be a superset of the actually used pipes, introducing ambiguity. Nevertheless we can make sure that the additionally guessed pipes form cycles and are thus irrelevant.

As an application consider the function $IP(x,y) = \sum_{i=1}^{n} (x_i\cdot y_i)$ $mod$ $2$. It is well known that  $N(IP) \geq n+1$ \cite{KN96}, hence we get that $GH(IP)\geq n$. The same bound holds for Disjointness. These bounds improve on the previous $\Omega(n/\log n)$ bounds for these functions \cite{BFSS11}. Furthermore note that the fooling set technique gives only bounds of size $O(\log^2 n)$ for the complexity of $IP$ (see \cite{KN96}), so the technique previously used to get a linear lower bound for Equality fails for $IP$.

\subsection{$GH(f)$ At Most The Size of a Protocol Tree for $f$}

Buhrman et al.~\cite{BFSS11} show that any one way communication complexity protocol with complexity $D_1(f)$ can be converted to a garden-hose protocol with $2^{D_1(f)}+1$ pipes. One-way communication complexity can be much larger than two-way communication \cite{papadimitriou&sipser:cc}.

\begin{theorem}\label{thm:dcc} For any function $f$, the garden-hose complexity $GH(f)$ is upper bounded by the number of edges in a protocol tree for $f$.
\end{theorem}

 The construction is better than the previous one in \cite{BFSS11} for problems for which one-way communication is far from the many-round communication complexity.

\section{Relating Permutation Branching Programs and the Garden-Hose Model}

\begin{definition} In a garden hose protocol a spilling-pipe on a player's side is a pipe such that water spills out of that pipe on the player's side during the computation for some input $x,y$.

 We say a protocol has multiple spilling-pipes if there is more than one spilling-pipe on Alice's side or on Bob's side.
\end{definition}

We now show a technical lemma that helps us compose garden-hose protocols without blowing up the size too much.

\begin{lemma}\label{one-output} A garden-hose protocol $P$ for $f$ with multiple spilling pipes can be converted to another garden-hose protocol $P'$ for $f$ that has only one spilling pipe on Alice's side and one spilling pipe on Bob's side. The size of $P'$ is at most 3 times the size of $P$ plus 1.
\end{lemma}

Next we are going to show that it is possible to convert a (loose) permutation branching program into a garden-hose protocol with only a constant factor increase in size. We are stating a more general fact, namely that the inputs to the branching program we simulate can be functions (with small garden-hose complexity) instead of just variables. This allows us to use composition.

\begin{lemma}\label{pbp2gh} GH$(g(f_1, f_2,..., f_k)) = O(s\cdot\max(C_i)) + O(1)  $, where $PBP(g) = s$ and $GH(f_i) = C_i$ and $f_i : \{0, 1 \}^n \times \{0, 1 \}^n \rightarrow \{0, 1 \}$. The $f_i$ do not necessarily have the same inputs $x,y$.
\end{lemma}

A first corollary is the following fact already shown in \cite{BFSS11}.
Nonuniform Logspace is equal to the class of all languages recognizable by polynomial size families of branching programs.
Since reversible Logspace equals deterministic Logspace \cite{KMT00}, and a reversible Logspace machine (on a fixed input length) can be  transformed into a polynomial size permutation branching program, we get the following.

\begin{corollary} \label{ncgh} Logspace $\subseteq GH(poly(n))$. This holds for any partition of the variables among Alice and Bob.\end{corollary}

\section{The Distributed Majority Function}

In this section we investigate the complexity of the Distributed Majority function.

\begin{definition} Distributed Majority: DMAJ$(x,y)=1$ iff $\sum_i^n(x_i \cdot y_i) \geq \frac{n}{2}$, where $x,y\in\{0,1\}^n$.
\end{definition}

Buhrman et al.~\cite{BFSS11} have conjectured that the complexity of this function is quadratic, which is what is suggested by the na\"{i}ve garden-hose protocol for the problem. The na\"{i}ve protocol implicitly keeps one counter for $i$ and one for the sum, leading to quadratic size.
Here we describe a construction of a permutation branching program of size $O(n \cdot \log^{3} n)$ for Majority, which can then be used to construct a garden-hose protocol for the Distributed Majority function. The Majority function is defined by $MAJ(x_1,\ldots, x_1)=1\Leftrightarrow \sum x_i\geq n/2$.

Note that the Majority function itself can be computed in the garden-hose model using $O(n)$ pipes (for any way to distribute inputs to Alice and Bob), since Alice can just communicate $\sum_i x_i$ to Bob.
The advantage of using a permutation branching program to compute Majority is that by Lemma \ref{pbp2gh} we can then find a garden-hose protocol for the composition of MAJ and the Boolean AND, which is the Distributed Majority function. We adapt a construction of Sinha and Thathachar \cite{ST97}, who describe a branching program for the Majority function.

\begin{lemma} \label{pbpmaj} $PBP(MAJ) = O(n\cdot \log^{3} n)$.  
\end{lemma}

We can now state our result about the composition of functions $f_1,\ldots, f_k$ with small garden-hose complexity via a Majority function.

\begin{lemma} \label{ghmaj} For $(f_1, f_2, .., f_k)$, where each function $f_i$ has garden-hose complexity  $GH(f_i)$, we have $GH(MAJ( f_1,\ldots, f_k)) = O(\sum GH(f_i))\cdot\log^3 k)$.
\end{lemma}

The lemma immediately follows from combining Lemma \ref{pbpmaj} with Lemma \ref{pbp2gh}. Considering $f_i=x_i\wedge y_i$ we get

\begin{corollary}
The garden-hose complexity of distributed Majority is $O(n\log^3 n)$.
\end{corollary}

\section{Composition and Connection to Formula Size}

We wish to relate $GH(f)$ to the formula size of $f$. To do so we examine composition of garden-hose protocols by popular gate functions.

\begin{theorem} \label{our-thm} For $(f_1, f_2, .., f_k)$, where each function $f_i$ has garden-hose complexity $GH(f_i)$
\begin{itemize}
 \item $GH(\bigvee f_i) = O(\sum GH(f_i))$.
 \item $GH(\bigwedge f_i) = O(\sum GH(f_i))$.
 \item $GH(\oplus f_i) = O(\sum GH(f_i))$.
 \item $GH(MAJ( f_i)) = O(\sum GH(f_i)\cdot \log^3 k)$.
\end{itemize}
\end{theorem}

This result follows from Lemma \ref{ghmaj} and Lemma \ref{pbp2gh} combined with the trivial loose permutation branching programs for AND, OR, XOR.

We now turn to the simulation of Boolean formulae by garden-hose protocols. We use the simulation of formulae over the set of all fan-in 2 function by branching programs due to Giel \cite{giel01}.

\begin{theorem}\label{thm:circ} Let $F$ be a formula for a Boolean function $g$ on $k$ inputs made of gates $\{\wedge, \vee,  \oplus \}$ of arbitrary fan-in. If $F$ has size $s$ and $GH(f_i)\leq c$ for all $i$, then for all constants $\epsilon>0$ we have $GH(g(f_1, f_2, .., f_k)) \leq O(s^{1+\epsilon}\cdot c)$.
\end{theorem}

\begin{proof}
Giel \cite{giel01} shows the following simulation result:

\begin{fact}
Let $\epsilon>0$ be any constant. Assume there is a formula with arbitrary fan-in 2 gates and size $s$ for a Boolean function $f$. Then there is a layered branching program of size $O(s^{1+\epsilon})$ and width $O(1)$ that also computes $f$.
\end{fact}

By inspection of the proof it becomes clear that the constructed branching program is in fact a strict permutation branching program.
The theorem follows by applying Lemma~\ref{pbp2gh}.
\end{proof}

\begin{corollary} When the $f_i$'s are single variables $GH(g) \leq O(s^{1+\epsilon})$ for all constants  $\epsilon>0$. Thus any lower bound on the garden-hose complexity of a function $g$ yields a slightly smaller lower bound on formula-size (all gates of fan-in 2 allowed).\end{corollary}

The best lower bound of $\Omega(n^2/\log n)$ known for the size of formulae over the basis of all fan-in 2 gate function is due to Ne\v ciporuk \cite{Ne66}. The Ne\v ciporuk lower bound method (based on counting subfunctions) can also be used to give the best general branching program lower bound of $\Omega(n^2/\log^2 n)$ (see \cite{wegener:complexity}).

Due to the above any lower bound larger than $\Omega(n^{2+\epsilon})$ for the garden-hose model would immediately give lower bounds of almost the same magnitude for formula size and permutation branching program size. Proving super-quadratic lower bounds in these models is a long-standing open problem.

Due to the fact that we have small permutation branching programs for Majority, we can even simulate a more general class of formulae involving a limited number of Majority gates.

\begin{theorem}
Let $F$ be a formula for a Boolean function $g$ on $n$ inputs made of gates $\{\wedge, \vee, \oplus \}$ of arbitrary fan-in.
Additionally there may be at most $O(1)$ Majority gates on any path from the root to the leaves.
 If $F$ has size $s$, then for all constants  $\epsilon>0$ we have
$GH(g) \leq O(s^{1+\epsilon})$.
\end{theorem}

\begin{proof}
Proceeding in reverse  topological order we can replace all sub-formulae below a Majority gate by garden-hose protocols with Theorem \ref{thm:circ}, increasing the size of the sub-formula. Then we can apply Lemma \ref{ghmaj} to replace the sub-formula including the Majority gate by a garden-hose protocol. If the size of the formula below the Majority gate is $\tilde{s}$, then the garden-hose size is $O(\tilde{s}^{1+\epsilon'})$, where the poly-logarithmic factor of Lemma \ref{ghmaj} is hidden in the polynomial increase.
Since every path from root to leaf has at most $c=O(1)$ Majority gates, and we may choose the $\epsilon'$ in Theorem \ref{thm:circ} to be smaller than $\epsilon/c$, we get our result.
\end{proof}

\subsection{The Ne\v ciporuk Bound with Arbitrary Symmetric Gates}

Since garden-hose protocols can even simulate formulae containing some arbitrary fan-in Majority gates, the question arises whether one can hope for super-linear lower bounds at all. Maybe it is hard to show super-linear lower bounds for formulae having Majority gates? Note that very small formulae for the Majority function itself are not known (the currently best construction yields formulae of size $O(n^{3.03})$ \cite{majorityform}), hence we cannot argue that Majority gates do not add power to the model.
In this subsection we sketch the simple observation that the Ne\v ciporuk method \cite{Ne66} can be used to give good lower bounds for formulae made of {\em arbitrary symmetric gates of any fan-in}. Hence there is no obstacle to near-quadratic lower bounds from the formula size connection we have shown. We stress that nevertheless we do not have any super-linear lower bounds for the garden-hose model.

We employ the communication complexity notation for the Ne\v ciporuk bound from \cite{kla:neci}.

\begin{theorem}\label{thm:neci}
Let $f:\{0,1\}^n\to\{0,1\}$ be a Boolean function and $B_1,\ldots, B_k$ a partition of the input bits of $f$. Denote by $D_j(f)$ the deterministic one-way communication complexity of $f$, when Alice receives all inputs except those in $B_j$, and Bob the inputs in $B_j$.
Then the size (number of leaves) of any formula consisting of arbitrary symmetric Boolean gates is at least $\sum D_j(f)/\log n$.
\end{theorem}

The theorem is as good as the usual Ne\v ciporuk bound except for the log-factor, and can hence be used to show lower bounds of up to $\Omega(n^2/\log^2 n)$ on the formula size of explicit functions like IndirectStorageAccess \cite{wegener:complexity}.

\section{Time Bounded Garden-Hose Protocols}

We now define a notion of time in garden-hose complexity.
\begin{definition}
 Given a garden-hose protocol $P$ for computing function $f$, and an input $x,y$ we refer to the pipes that carry water in $P$ on $x,y$ as the wet pipes. Let $T_P$ denote the maximum number of wet pipes for any input $(x,y)$ in $P$.
\end{definition}

The number of wet pipes on input $x,y$ is equal to the length of the path the water takes and thus corresponds to the time the computation takes.
Thus it makes sense to investigate protocols which have bounded time $T_P$. Furthermore, the question is whether it is possible to simultaneously optimize $T_P$ and the number of pipes used.

\begin{definition}
We define $GH_k(f)$ to be the complexity of an optimal garden-hose protocol $P$ for computing $f$ where for any input $(x,y)$ we have that $T_P$ is bounded by $k$.
\end{definition}

As an example consider the Equality function (test whether $x=y$). The straightforward protocol that compares bit after bit has cost $3n$ but needs time $2n$ in the worst case. On the other hand one can easily obtain a protocol with time 2, that has cost $O(2^n)$: use $2^n$ pipes to communicate $x$ to Bob.
We have the following general lower bound.

\begin{theorem}\label{thm:simround}
For all Boolean functions $f$ we have $GH_k(f) = \Omega(2^{D_k(f)/k})$,
where $D_k(f)$ is the deterministic communication complexity of $f$ with at most $k$ rounds (Alice starting).
\end{theorem}
\begin{proof}
We rewrite the claim as $D_k(f)=O(k \cdot \log GH_k(f))$.

Let $P'$ be the garden-hose protocol for $f$ that achieves complexity $GH_k(f)$ for $f$.
The deterministic $k$-round communication protocol for $f$ simulates $P'$ by simply following the flow of the water. In each round Alice or Bob (alternatingly) send the name of the pipe used at that time by $P'$.
\end{proof}

Thus for Equality we have for instance that $GH_{\sqrt{n}}(Equality)=\Omega(2^{\sqrt{n}})$.
There is an almost matching upper bound of $GH_{\sqrt{n}}(Equality) = O(2^{\sqrt{n}} \cdot \sqrt{n})$ by using $\sqrt n$ blocks of $2^{\sqrt n}$ pipes to communicate blocks of $\sqrt n$ bits each.

We can easily deduce a time-cost tradeoff from the above: For Equality the product of time and cost is at least $\Omega(n^2 / \log n)$, because for time $T<o(n/\log n)$ we get a super-linear bound on the size, whereas for larger $T$ we can use that the size is always at least $n$.

\subsection{A Time-Size Hierarchy}

The Pointer Jumping Function is well-studied  in communication complexity. We describe a slight restriction of the problem in which the inputs are permutations of $\{1,\ldots, n\}$.

\begin{definition}
 Let $U$ and $V$ be two disjoint sets of vertices such that $|U| = |V| = n$.

 Let $F_A = \{ f_A | f_A: U \rightarrow V$ and $f_A$ is bijective$\}$ and $F_B = \{ f_B | f_B: V \rightarrow U$  and $f_B$ is bijective$\}$.
For a pair of functions $f_A\in F_A$ and $f_B\in F_B$ define
  $f(v) = \left\{
  \begin{array}{l l}
    f_A(v) & \quad \text{if $v \in U$}\\
    f_B(v) & \quad \text{if $v \in V$.}
  \end{array} \right.$

Then   $f_0(v)=v$ and $f_k(v)=f(f_{k-1}(v))$.

  Finally, the pointer jumping function $PJ_k : F_A \times F_B \rightarrow \{ 0, 1\}$ is defined to be the XOR of all bits in the binary name of $f_k(v_0)$, where $v_0$ is a fixed vertex in $U$.
\end{definition}

Round-communication hierarchies for $PJ_k$ or related functions are investigated in \cite{NW93}. Here we observe that $PJ_k$ gives a time-size hierarchy in the garden-hose model. For simplicity we only consider the case where Alice starts.

\begin{theorem}\label{pj}
\begin{enumerate}
\item $PJ_k$ can be computed by a garden-hose protocol with time $k$ and size $kn$.
\item Any garden-hose protocol for $PJ_k$ that uses time at most $k-1$ has size $2^{\Omega(n/k)}$ for all $k\leq n/(100\log n)$.
\end{enumerate}\end{theorem}

We note that slightly weaker lower bounds hold for the randomized setting.

\section{Randomized Garden-Hose Protocols}
We now bring randomness into the picture and investigate its power in the garden-hose model. Buhrman et al~\cite{BFSS11} have already considered protocols with public randomness. In this section we are mainly interested in the power of private randomness.

\begin{definition}
 Let $RGH^{pub}(f)$ denote the minimum complexity of a garden-hose protocol for computing $f$, where the players have access to public randomness, and the output is correct with probability 2/3 (over the randomness).
Similarly, we can define $RGH^{pri}(f)$, the cost of garden-hose protocols with access to private randomness.
\end{definition}

By standard fingerprinting ideas \cite{KN96} we can observe the following.

\begin{claim}
 $RGH^{pub}(Equality) = O(1)$
\end{claim}

\begin{claim}\label{REQn}
 $RGH^{pri}(Equality) = O(n)$, and this is achieved by a constant time protocol.
\end{claim}
\begin{proof}
The second claim follows from Newman's theorem {\cite{newman:random}} showing that any public coin protocol with communication cost $c$ can be converted into a private coin protocol with communication cost $c+ \log n +O(1)$ bits on inputs of length $n$ together with the standard public coin protocol for Equality, and the protocol tree simulation of Theorem \ref{thm:dcc}.
\end{proof}

Of course we already know that even the deterministic complexity of Equality is $O(n)$, hence the only thing achieved by the above protocol is the reduction in time complexity. Note that due to our result of the previous section computing Equality deterministically in constant time needs exponentially many pipes.

Buhrman et al. {\cite{BFSS11}} have shown how to de-randomize a public coin protocols at the cost of increasing size by a factor of $O(n)$, so the factor $n$ in the separation between public coin and deterministic protocols above is the best that can be achieved. This raises the question whether private coin protocols can ever be more efficient in size than the optimal deterministic protocol.
We now show that there are no very efficient private coin protocols for Equality.

\begin{claim}
 $RGH^{pri}(Equality) = \Omega(\sqrt{n}/\log n)$
\end{claim}
\begin{proof}
To prove this we first note that $RGH^{pri}(f) = \Omega(R^{||}(f) / \log R^{||}(f))$, where $R^{||}(f)$ is the cost of randomized private coin simultaneous message protocols for $f$ (Alice and Bob can send their connections to the referee).
Hence, $RGH^{pri}(f) = \Omega(R^{|| pri}(f) / \log R^{|| pri}(f))$, but Newman and Szegedy \cite{New96} show that $RGH^{pri}(Equality) =\Omega(\sqrt{n})$.
\end{proof}

\section{Open Problems}
\begin{itemize}
\item We show that getting lower bounds on $GH(f)$ larger than $\Omega(n^{2+\epsilon})$ will be hard. But we know of no obstacles to proving super-linear lower bounds.
\item Possible candidates for quadratic lower bounds could be the Disjointness function with set size $n$ and universe size $n^2$, and the IndirectStorageAccess function.
\item Consider the garden-hose matrix $G_s$ as a communication matrix. How many distinct rows does $G_s$ have? What is the deterministic communication complexity of $G_s$? The best upper bound is $O(s\log s)$, and the lower bound is $\Omega(s)$. An improved lower bound would give a problem, for which $D(f)$ is larger than $GH(f)$.
\item We have proved $RGH^{pri}(Equality) = \Omega(\sqrt{n}/\log n)$. Is it true that $RGH^{pri}(Equality) = \Theta(n)$? Is there any problem where $RGH^{pri}(f)$ is smaller than $GH(f)$?
\item It would be interesting to investigate the relation between the garden-hose model and memoryless communication complexity, i.e., a model in which Alice and Bob must send messages depending on their input and the message just received only. The garden-hose model is memoryless, but also reversible.
\end{itemize}

\section*{Acknowledgement}

We thank an anonymous referee for pointing out a mistake in an earlier version of this paper.
\bibliographystyle{plain}
\bibliography{bibo}

\begin{thebibliography}{10}

\bibitem{bar85}
D.A. Barrington.
\newblock Width-3 permutation branching programs, 1985.
\newblock Technical report, MIT/LCS/TM-293.

\bibitem{beame:commtrade}
P.~Beame, M.~Tompa, and P.~Yan.
\newblock Communication-space tradeoffs for unrestricted protocols.
\newblock {\em SIAM Journal on Computing}, 23(3):652--661, 1994.
\newblock Earlier version in FOCS'90.

\bibitem{BrodyCPSS13}
Joshua Brody, Shiteng Chen, Periklis~A. Papakonstantinou, Hao Song, and
  Xiaoming Sun.
\newblock Space-bounded communication complexity.
\newblock In {\em Proceedings of the 4th conference on Innovations in
  Theoretical Computer Science}, pages 159--172, 2013.

\bibitem{BFSS11}
Harry Buhrman, Serge Fehr, Christian Schaffner, and Florian Speelman.
\newblock The garden-hose model.
\newblock In {\em Proceedings of the 4th conference on Innovations in
  Theoretical Computer Science}, pages 145--158. ACM, 2013.

\bibitem{Ch13}
Well~Y Chiu, Mario Szegedy, Chengu Wang, and Yixin Xu.
\newblock The garden hose complexity for the equality function.
\newblock {\em arXiv:1312.7222}, 2013.

\bibitem{giel01}
O.~Giel.
\newblock Branching program size is almost linear in formula size.
\newblock {\em Journal of Computer and System Sciences}, 63(2):222--235, 2001.

\bibitem{klauck:tradeoffs}
H.~Klauck.
\newblock Quantum and classical communication-space tradeoffs from rectangle
  bounds.
\newblock In {\em Proceedings of FSTTCS}, 2004.

\bibitem{kla:neci}
H.~Klauck.
\newblock {One-Way Communication Complexity and the Ne\v ciporuk Lower Bound on
  Formula Size}.
\newblock {\em SIAM J. Comput.}, 37(2):552--583, 2007.

\bibitem{ksw:dpt-siam}
H.~Klauck, R.~{\v{S}}palek, and R.~{de} Wolf.
\newblock Quantum and classical strong direct product theorems and optimal
  time-space tradeoffs.
\newblock {\em SIAM Journal on Computing}, 36(5):1472--1493, 2007.
\newblock Earlier version in FOCS'04. quant-ph/0402123.

\bibitem{KN96}
Eyal Kushilevitz and Noam Nisan.
\newblock {\em Communication Complexity}.
\newblock Cambridge University Press, 1997.

\bibitem{KMT00}
K.J. Lange, P.~McKenzie, and A.~Tapp.
\newblock Reversible space equals deterministic space.
\newblock {\em Journal of Computer and System Sciences}, 2(60):354--367, 2000.

\bibitem{Ne66}
E.~I. Ne\v{c}iporuk.
\newblock A boolean function.
\newblock In {\em Soviet Mathematics Doklady}, volume~7, 1966.

\bibitem{newman:random}
I.~Newman.
\newblock Private vs.~common random bits in communication complexity.
\newblock {\em Information Processing Letters}, 39(2):67--71, 1991.

\bibitem{New96}
Ilan Newman and Mario Szegedy.
\newblock Public vs. private coin flips in one round communication games
  (extended abstract).
\newblock In {\em Proceedings of the Twenty-eighth Annual ACM Symposium on
  Theory of Computing}, STOC '96, pages 561--570, 1996.

\bibitem{NW93}
Noam Nisan and Avi Wigderson.
\newblock Rounds in communication complexity revisited.
\newblock {\em SIAM J. Comput.}, 22(1):211--219, February 1993.

\bibitem{papadimitriou&sipser:cc}
C.~H. Papadimitriou and M.~Sipser.
\newblock Communication complexity.
\newblock {\em Journal of Computer and System Sciences}, 28(2):260--269, 1984.
\newblock Earlier version in STOC'82.

\bibitem{Overlays}
P.~Papakonstantinou, D.~Scheder, and H.~Song.
\newblock Overlays and limited memory communication mode(l)s.
\newblock In {\em Proc. of the 29th Conference on Computational Complexity},
  2014.

\bibitem{majorityform}
I.~S. Sergeev.
\newblock Upper bounds for the formula size of symmetric boolean functions.
\newblock {\em Russian Mathematics, Iz. VUZ}, 58(5):30--42, 2014.

\bibitem{ST97}
Rakesh~Kumar Sinha and Jayram~S Thathachar.
\newblock Efficient oblivious branching programs for threshold and mod
  functions.
\newblock {\em Journal of Computer and System Sciences}, 55(3):373--384, 1997.

\bibitem{wegener:complexity}
I.~Wegener.
\newblock {\em The Complexity of {B}oolean Functions}.
\newblock Wiley-Teubner Series in Computer Science, 1987.

\end{thebibliography}

\section{Appendix}

\subsection{Non-deterministic Communication}

\begin{proof}[Proof of Theorem \ref{thm:ncc}]
Consider a deterministic garden-hose protocol $P$ for $f$ using $s$ pipes.
Maybe the most natural approach to simulate $P$'s computation by a non-deterministic communication protocol would be to guess the path that the water takes, and verify this guess locally by Alice and Bob. There are, however, too many paths for this to lead to good bounds. Instead we use a coarser guess.
For any given input $x,y$ in a computation of $P$ the water traverses a set $W(x,y)$ of pipes. We refer to these pipes as the {\em wet} pipes in $P$ on $x,y$. In general a set of wet pipes can correspond to several paths through the network, which must use only edges from the set.

In the non-deterministic protocol Alice guesses a set $S$ of pipes that is supposed to be $W(x,y)$.
Since $|W(x,u)|$ is odd if and only if $f(x,y)=1$ the size of $S$ immediately tells us whether $S$ is a witness for 1-inputs or 0-inputs.

Consider an even size set $S$.
Alice computes the connections of the pipes on her side using her input $x$ (as used in the garden-hose protocol).
Her connections are {\em consistent} with $S$, iff the tap is connected to a pipe in $S$, and the other pipes in $S$ are all connected in pairs, except one, which is open. Note that none of the pipes in $S$ may be connected to a pipe outside of $S$.
Similarly, $S$ is consistent with Bob's connections (based on $y$), if all the pipes in $S$ are paired up (no pipe in $S$ is open and no pipe in $S$ is connected to a pipe outside $S$).

For odd size $S$ we use an analogous definition of consistency: Now Alice has no open pipe in $S$ and all pipes in $S$ are paired up except the one connected to the tap, and Bob has all pipes in $S$ paired up except one that is open.

Suppose that $S$ is consistent with the connections defined by $x,y$.
Denote by $P(x,y)$ the path the water takes in the garden-hose protocol. We claim that all the pipes in $P(x,y)$ are in $S$, and that the remaining pipes in $S$ form cycles. If this is the case then the non-deterministic protocol is correct: Since cycles have even length, subtracting them does not change the fact that $|S|$ is even or odd, and hence the size of $S$ and $P(x,y)$ have the same parity, i.e., a consistent $S$ determines the function value correctly. Also note that the communication complexity of the non-deterministic protocol is at most $GH(f)$+1, since a subset of the pipes used can be communicated with $s$ bits: Alice guesses an $S$ that is consistent with her input and sends it to Bob, who accepts/rejects if $S$ is also consistent with his input, otherwise he gives up (accepting/rejecting takes one additional bit of communication). Note that for partial functions no consistent $S$ may exist for Alice to choose, but in that case she can give up without a result.

To establish correctness we have to show that all pipes in $P(x,y)$ are in $S$ (and the remaining pipes in $S$ form cycles). Clearly the starting pipe (the one connected to the tap) is in $S$ by the definition of consistency. All remaining pipes in $S$ on Bob's and Alice's side are either paired up or (for exactly one pipe) open. Hence we can follow the flow of water without leaving $S$. This implies that $P(x,y)$ is in $S$, and since removing $P(x,y)$ from $S$ leaves no open pipes all the remaining pipes in $S$ must form a set of cycles.
\end{proof}

\subsection{Garden-Hose and Protocol Trees}

\begin{proof}[Proof of Theorem \ref{thm:dcc}:]

Given a protocol tree (with $k$ edges) of a two way communication protocol $P$ for any function $f$ we construct a garden-hose protocol with at most $k$ pipes.

We describe the construction in a recursive way. Let $v$ be any node of the protocol tree belonging to Alice, with children $u_1,\ldots, u_d$ belonging to Bob. In the protocol tree rooted at $v$ a function $f_v$ is computed. If none of the $u_i$ are leaves, then we assume by induction that we can construct a garden hose protocol $P_i$ for each of the children, where $P_i$ uses at most $s_i$ many pipes, and $s_i$ is the number of edges in the subtree of $u_i$. The $P_i$ have the tap on Bob's side. To find a garden-hose protocol for $v$, we use $d+\sum s_i$ pipes. Alice sends the water through pipe $i$ to communicate the message corresponding to the edge to $u_i$. Furthermore the right end of pipe $i$ is connected to the tap of a copy of $P_i$. The number of pipes used $(d+\sum s_i)$ is at most the number of edges in the protocol tree. If one or two of the $u_i$ are leaves, we use the same construction, except that for an accepting leaf we use one extra pipe that is open on Bob's end, and for a rejecting leaf we just let the water spill at Alice's pipe. It is easy to see by induction that the garden-hose protocol accepts on $x,y$ if and only if the protocol tree ends in an accepting leaf.
\end{proof}

\subsection{One Spilling Pipe}

\begin{proof}[Proof of Lemma \ref{one-output}]
 Fix a protocol $P$ that uses $s$ pipes to compute $f$. In the protocol $P$ Alice makes the connections on her side based on her input $x$. Similarly Bob's connections are based on his input $y$.
Denote the set of pipes that are open on Alice's side by $S_A$ and the set of  pipes that are open on Bob's side by $S_B$.

 In the new protocol $P'$ Alice and Bob have $3s$ pipes arranged into 3 blocks of $s$ pipes each. Let's call them $B_1,B_2$ and $B_3$.
 The main idea is to use $B_1$ to compute $f$ and then use $B_2$ and $B_3$ to `un-compute' $f$ (to remove the extra information provided by the multiple spilling pipes).

 In the construction of $P'$ Alice and Bob make their connections on $B_1,B_2$ and $B_3$ separately, exactly the same way they did in $P$ for $s$ pipes. Alice then connects $B_1$'s tap-pipe to the tap and keeps the tap-pipes of $B_2, B_3$ open. They then add the following connections: Alice connects every pipe $i\in S_A$ in $B_1$ to pipe $i\in S_A$ in $B_2$ and Bob connects every pipe $i\in S_B$ in $B_1$ to pipe $i\in S_B$ in $B_3$. Note that those pipes were open before they were connected as they were all spilling pipes. $B_1$ now does not have any open pipes. The only pipes that will ever spill in $B_2$ and $B_3$ are their taps (there may be other open pipes but it is easy to see that they never spill). The tap-pipes of $B_2$ and $B_3$ are both on Alice's side. Finally, Alice uses one more pipe, and connects the tap-pipe of $B_3$ to the new pipe. Figure 1 shows an example of the construction.

\begin{figure}
\begin{center}
 \includegraphics[scale=0.5]{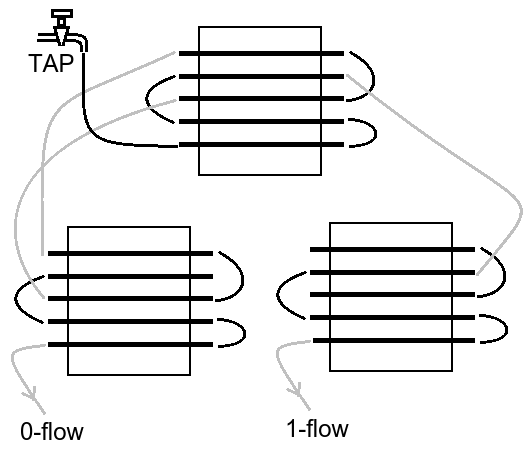}
 \caption{The Construction in Lemma \ref{one-output}}
\end{center}
\end{figure}

The size of the new protocol $P'$ is exactly $3s+1$, and there is exactly one spilling pipe on each side, namely the tap pipes of $B_2$ and $B_3$, because the only other open pipes are the $S_A$ pipes in $B_3$ and the $S_B$ pipes in $B_2$. These cannot be reached by the water. All connections made are done by Alice and Bob alone. We now argue that the protocol computes $f(x,y)$ correctly.

Notice that if $f(x,y)=0$, then water flows through $B_1$ and ends at one of the pipes in $S_A$. This pipe is connected to the corresponding pipe in $B_2$. So the water follows the same path backwards in $B_2$ until it reaches the tap-pipe in $B_2$. This pipe is open on Alice's side. Hence water spills on Alice's side making the output 0 (and it spills at the tap of $B_2$).

Similarly, if $f(x,y)=1$, water flows through $B_1$ and ends at one of the pipes in $S_B$ on Bob's side. Since this pipe is connected to the corresponding pipe in $B_3$ the water flows backwards din $B_3$ until it reaches the tap-pipe of $B_3$. This is on Alice's side and connected to the extra pipe. This makes the water to spill on Bob's side as desired.
\end{proof}

\subsection{Permutation Branching Programs to Garden-Hose}

\begin{proof}[Proof of Lemma \ref{pbp2gh}]
 In Lemma \ref{one-output} we have seen that we can turn a garden-hose protocol with multiple spilling pipes into a protocol with exactly one spilling pipe per side. Such a protocol acts exactly as a node in a branching program, except that its decision is based on $f_i(x_i,y_i)$. This observation suffices to simulate decision trees, but in a branching program nodes can have in-degree larger than 1, and we cannot pump water from several sources into a single garden-hose protocol.

We now show how to construct a garden-hose protocol for $g(f_1, f_2,..., f_k)$.
Given a loose permutation branching program for $g$  of size $S$, we show how to construct a garden-hose protocol.

Let $G$ denote the graph of the branching program. $G$ consists of $T$ layers $L_0,\ldots, L_{T-1}$, where the first layer has just one node (the source), the last layer 2 nodes (the sinks), and all intermediate layers have $W$ nodes, so the size is $S=(T-2)W+3$.
Layer $L_i$ queries some variable $z_i$, whose value is $f_i(x_i,y_i)$. The 1-edges between $L_i$ and $L_{i+1}$ are $E^1_i$, the 0-edges $E^0_i$.

The construction goes by replacing the nodes of each layer by the garden-hose protocols $P_i$ for $f_i$.
Each layer uses $2W$ copies of $P_i$, arranged in two layers. We refer to these copies as the upper and lower copies of $P_i$, each numbered from $1$ to $W$ (and implicitly by their level). Essentially we need the first layer to compute $f_i$, and the second layer to un-compute, since we only want to remember the name of the current vertex in $G$, not the value of $f_i$.

 If $e=(j,k)\in E_i^1$, then we connect the 1-spill pipe of the upper $j$-th copy of $P_i$ to the 1-spill pipe of the lower $k$-th copy of $P_{i}$. Similarly we make the connections for the 0-spill pipes (on Alice's side).

 To connect layers we connect the tap-pipes on each lower copy $j$ of a level $L_i$ to the tap-pipes of an upper copy $j$ on level $L_{i+1}$. On level $L_0$ the tap-pipe of an upper copy is connected to Alice's tap according to the branching program.

\begin{figure}
\begin{center}
 \includegraphics[scale=0.5]{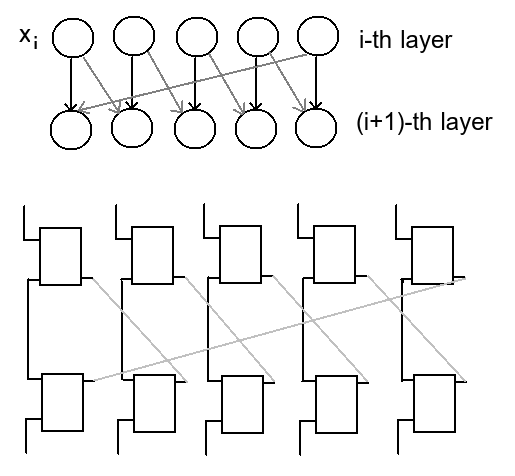}
 \caption{Permutation Branching Program to Garden-Hose Protocol Construction}\label{fig:pbp}
\end{center}
\end{figure}

 Figure \ref{fig:pbp} shows an example of the construction, where each block is a garden-hose protocol to compute $f_i$.

For every edge that goes to the accepting sink of the branching program we use one pipe that is connected to the corresponding upper copy
on Alice's side, if the corresponding spilling pipe is on Alice's side. Otherwise we leave the spilling pipe open. We proceed analogously for edges to the rejecting sink.

The size of the garden-hose protocol is at most $2W\cdot\sum_{i=1}^L C_i\leq \max C_i\cdot 2WL$.
\end{proof}

\subsection{A Permutation Branching Program for Majority}

\begin{proof}[Proof of Lemma \ref{pbpmaj}]

In 1997, Sinha et al.~\cite{ST97} described a branching program of size $O(\frac{n \log^3 n}{\log\log n \log\log\log n})$ for computing Majority. Unfortunately the branching program they construct is not a permutation branching program.
Thus it is not immediately clear how to convert their construction into a garden-hose protocol.

To describe a permutation branching program for Majority we first need permutation branching programs for computing the sum of the inputs mod $r$ for small $r$. Denote by $Mod_r$ the (non-Boolean) function $Mod_r(x_1,\ldots, x_n)=\sum_i x_i$ mod $r$. The following is easy to see.

\begin{claim} $Mod_r(x_1,\ldots, x_n)$ can be computed by permutation branching program of width $r$ so that each input $x$ with $|x|=i$, when starting on the top level at node $j$ ends at node $i+j$ mod $r$ on the last level.\end{claim}

We call this permutation branching program a modulus-$r$ box.
The {\em join} of two modulus $r_1$ resp.~$r_2$ boxes is a new branching program, in which bottom level nodes of the first box are identified in some way with top level nodes of the second. We employ the following main technical result of Sinha et al.~\cite{ST97}, which describes an approximate divider.

\begin{fact} \label{sinha} \cite{ST97} Fix the length $M$ of an interval of natural numbers.
 There are $k\leq \log M$ prime numbers $r_2 < r_3 < \cdots < r_{k}$, where $r_2>4\log M$ and $r_k<12\log M$ and a number $r_1=2^t$ such that
 and $2M\leq \Pi_{1\leq i \leq k} r_i$ and $r_1<r_2$. Set $M'=M/r_2$. Consider inputs $x$ such that $b\leq |x|\leq b+M-1$.

Then there is a way to join $k$ modulus-$r_i$ boxes (in order $r_1\ldots, r_k$) into a single branching program, such that all inputs $x$ reaching the sink nodes named $lM'$ mod $r_k$ for some $0\leq l<r_k$ (in the last box) satisfy that $|x|$ belongs to one of $r_2$ intervals of length $(k-1)M'$ in $[b,b+M-1]$. The intervals overlap, and each point in $[b,b+M-1]$ is in $k-1$ intervals.

Furthermore, the connections between the boxes are such that every output node of the $r_i$ box is connected to one input node of the $r_{i+1}$ box, and every input node of the $r_{i+1}$ box is connected to at most one output node of the $r_i$ box.
\end{fact}

The above differs from the presentation in \cite{ST97} in that we require that the $r_i$ are increasing so that we can join them without creating nodes with fan-in larger than 1. This means that every $r_i$ box for $i>1$ has a few input nodes that are not used.

Note that our goal is to know whether $|x|$ is greater than $n/2$ or not. Effectively this means there are three kinds of bottom layer nodes in the branching program constructed above (for $b=0$): those where we know that all inputs reaching the sink have $n/2>|x|$, at which point we can reject, those where $n/2<|x|$, where we accept, and undecided nodes. A bottom layer node is undecided, if the interval of possible $|x|$ reaching that sink contains $n/2$. At undecided nodes the interval of possible values of $|x|$ has been reduced to size $(k-1)M/r_2$, i.e., a $(k-1)/r_2<1/4$ fraction of the the original interval. Furthermore, there are $k-1$ undecided nodes (since $n/2$ is in that many intervals), but the intervals for those nodes stretch to at most $(k-1) M/r_2$ beyond $n/2$ on both sides, hence
the union of the intervals of all undecided bottom layer nodes is an interval of size at most $2(k-1)M/r_2\leq M/2$.
Hence, this construction can be iterated (at most $\log n$ times) to decide Majority on all inputs.

Now we need to argue that the whole construction can be made into a permutation branching program.
Obviously any mod-$r$ box can be computed by a strict permutation BP of width $r$ and length $n$. The connections between the $k$ boxed are injective mappings. Hence  so the whole constructions for the above fact can be made into a permutation branching program, where dummy nodes need to be added to bring all layers to the same width ($r_k$).

The branching program for Majority is then an iteration of the above construction of permutation branching programs. In each level of the iteration some nodes accept, some reject, and some continue on a smaller interval. For all undecided sink nodes we can assume that they continue using the same interval of size at most $M/2$. This continues until the intervals are very short ($M\leq\log n$), at which the problem can be solved by counting.

  To do the same iteration in a permutation branching program we need to do the following. We want to turn a building block $B_i$ of the iteration (a permutation branching program as in Fact \ref{sinha}) into a permutation branching program that has only 3 sinks reached by inputs (plus some sinks that are never reached). To do this we first use the original program, followed by 3 copies of the same program in reverse. We connect the undecided sinks of the upper program into the corresponding vertices in the first reversed lower program, similarly the accepting and rejecting sinks into the corresponding vertices of the other two reversed programs. Then each input that is undecided by $B_i$ will end up at the node corresponding to the starting node of the first reverse copy. Similarly inputs that are accepted by $B_i$ will leave the second reverse copy at the node corresponding to the starting node of $B_i$ etc. Using dummy nodes this program can be extended to a permutation branching program, with width increased by a factor of 3 and length by 2. Each input leads to one of three nodes. We can now connect the undecided sink of the above construction to the starting vertex of the next block $B_{i+1}$. To turn the whole construction into a strict permutation branching program the accepting and rejecting bottom vertices are connected to $O(n\log^2 n) $ extra vertices that remember at which layer/vertex the inputs were accepted/rejected.

  The whole construction yields a permutation branching program for Majority. The length of the program is $O(\log n\cdot \log n\cdot n)$, for the $\log n$ iterations, the $k\leq \log n$ boxes that have length $n$. Each level of the program has width at most $O(\log M)$ for the mod $r_i$ boxes and the constant factors to turn things into a permutation BP (plus $2\log n$ vertices for accepting/rejecting paths)). Hence the total size of the program is $O(n\log^3 n)$.

\end{proof}

\subsection{Lower Bound for Formulae with Symmetric Gates}

\begin{proof}[Proof of Theorem \ref{thm:neci}]
Fix $f$ and $B_1,\ldots, B_k$ and any formula $F$ of size at most $n^2$ computing $f$ consisting of symmetric gates only. Define $F_j$ to be the subtree of $F$, whose leaves are the variables in $B_j$ (and root is the output gate of the formula, and denote by $L_j$ the number of leaves of $F_j$. Then the size of $F$ is $\sum L_j$. We will show that $D_j(f)\leq O(L_j\cdot \log n)$.

Alice has all the variables except those in $B_j$, which go to Bob. Alice (and Bob) have to evaluate all the gates in $F_j$ (this includes the root).
They will evaluate the gates in (reverse) topological order. All the leaves are known to Bob. Denote by $P$ the set of paths in $F_j$ that start at a leaf or a gate of fan-in at least 2 inside $F_j$, and end at a gate of fan-in at least 2 inside $F_j$ and have no such gates in between. Then $L_j\geq |P|$. Also denote by $G$ the set of gates in $F_j$ that have fan-in larger than 1 inside $F_j$, again $L_j\geq |G|$. We will show that the communication is at most $O(L_j\cdot\log n)$.

Bob goes over paths $p\in P$ and gates in $g\in G$ in reverse topological order (i.e., from the leaves up).
Let $p=v_1,\ldots, v_t$ be the vertices of some $p$ in reverse topological order (i.e., the root is last). Denote by $f_p$ the gate at $v_{t-1}$, the last vertex that has fan-in 1 in $p$. Alice can tell Bob which function is computed at $v_{t-1}$ in terms of the value already computed (by Bob) at $v_1$. This takes 2 bits. Hence the total communication to evaluate paths in $P$ is $2|P|$. For each $g\in G$ there are at least 2 inputs in $F_j$ that have already been computed by Bob. Since the gate at $g$ is symmetric, it is sufficient for Alice to say how many of her inputs to $g$ evaluate to 1, which takes at most $2\log n$ bits unless the formula is larger than $n^2$. So the total communication is at most $O(|P|+|G|\log n)$, and $|G|,|P|\leq L_j$, unless $F$ has size larger than $n^2$ already.

\end{proof}

\subsection{Pointer Jumping}

\begin{proof}[Proof Sketch for Theorem \ref{pj}]
To show part 1) we use a protocol using $nk$ pipes, organized into $k$ blocks. If Alice has input $f_A$, then she connects the tap to pipe $f_A(v_1)$ in block 1. For all even numbered blocks $2j$ she connects the $i$th pipe in block $2j$ to pipe $f_A(i)$ in block $2j+1$.
Bob connects for all odd numbered blocks the $i$th pipe in block $2j+1$ to pipe $f_B(i)$ in block $2j+2$.

Assume that $k$ is odd. Then the $k$th vertex of the path is on Bob's side. If $PJ_k(f_A,f_B)=0$ then the XOR of $f_k(v_0)$ is 0 and the water needs to spill on Alice's side. Hence, in block $k$, for all pipes $i$ with even $i$, Alice leaves the pipe open instead of connecting it to a pipe in block $k-1$. She does make the connections as described above for all odd pipes in block $k$.

Similarly, if $k$ is even, then the last vertex is on Alice's side and if $f_k(v_0)$ is odd the spill needs to be on Bob's side. Hence Bob skips all the connections between blocks $k-1$ and $k$ for odd numbered pipes $i$ in block $k$.

Note that $f_A$ and $f_B$ are bijective, hence the connections made are legal.
In total we use $kn$ pipes. It is clear that the garden-hose protocol described above computes $PJ_k$.

Now we turn to part 2. Take any time $k-1$ garden-hose protocol for $PJ_k$ using $s$ pipes. Due to the simulation in Theorem \ref{thm:simround} we get a $k-1$ round communication protocol (Alice starting) with communication $(k-1)\log s$. But Nisan and Wigderson \cite{NW93} show that such protocols need communication $\Omega(n)$ for $k\leq n/(100\log n)$. Hence $s\geq2^{\Omega(n/k)}$.

The difficulty in applying their result is that Nisan and Wigderson analyze the complexity of $PJ_k$ for uniformly random inputs, not random {\em bijective} inputs $f_A$ resp. $f_B$. Hence we need to make some changes to their proof.
 These changes needed to make the argument work are minor, however: the uniform distribution on pairs of bijective functions is still a product distribution, and as long as $k=o(n)$ it is still true that at any vertex in the protocol tree the information about the next pointer is a small constant. The main difference to the original argument is that conditioning on the previous path introduces information about the next pointer due to the fact that vertices on the path can not be used again. This can easily be subsumed into the information given via the previous communication.
\end{proof}
\end{document}